\documentclass[a4paper,10pt,abstracton]{scrartcl}

\usepackage[latin1]{inputenc} 

\usepackage{amsmath}
\usepackage{amsfonts}
\usepackage{amssymb}
\usepackage{dsfont}

\usepackage{authblk}

\usepackage{amsthm}
\usepackage{mathtools}

\DeclareMathOperator{\im}{Im}
\renewcommand{\Im}{\im}
\newcommand{\ii}{\mathfrak{i}}
\newcommand{\LL}{\mathrm{L}}
\providecommand\given{}
\newcommand\SetSymbol[1][]{ 
  \nonscript\: #1 \vert \nonscript\:
  \mathopen{} 
  \allowbreak 
} 
\DeclarePairedDelimiterX\Set[1]{\{}{\}}{
  \renewcommand\given{\SetSymbol[\delimsize]}
  #1 
}

\let\epsilon\varepsilon

\usepackage[draft]{fixme}

\usepackage{graphicx}
\usepackage{subfig}
\usepackage{hyperref}

\theoremstyle{theorem}
\newtheorem{thm}{Theorem}[section]
\newtheorem{lem}[thm]{Lemma}
\newtheorem{cor}[thm]{Corollary}
\newtheorem{prp}[thm]{Proposition}
\newtheorem{cond}[thm]{Condition}

\theoremstyle{definition}
\newtheorem{rem}[thm]{Remark}

\DeclareMathOperator{\ad}{ad}

\DeclareMathOperator{\mdsr}{\mathds{R}}

\DeclareMathOperator{\ci}{\ii}

\DeclareMathOperator{\dd}{d}
\DeclareMathOperator{\Ei}{Ei}
\DeclareMathOperator{\Div}{div}

\renewcommand{\div}{\Div}

\providecommand{\abs}[2][]{#1\lvert#2#1\rvert}
\providecommand{\norm}[2][]{#1\lVert#2#1\rVert}
\providecommand{\jnorm}[2][]{#1\langle#2#1\rangle}
\providecommand{\inv}{^{-1}}

\newcommand{\R}{\mathds{R}}
\newcommand{\C}{\mathds{C}}

\newcommand{\N}{\mathds{N}}
\newcommand{\D}{\mathds{D}}
\newcommand{\id}{\mathds{1}}

\newcommand{\cH}{\mathcal{H}}

\begin{document}
\title{Spectral deformation for two-body dispersive systems with e.g.\
  the Yukawa potential} 
\author{Matthias Engelmann}
\affil{
  Fachbereich Mathematik\\
  Universit\"at Stuttgart} \author{Morten Grud Rasmussen}
\affil{Department of Mathematical Sciences\\
  Aalborg University}
\maketitle
\begin{abstract}
  We find an explicit closed formula for the $k$'th iterated
  commutator $\ad_A^k(H_V(\xi))$ of arbitrary order $k\ge1$ between a
  Hamiltonian $H_V(\xi)=M_{\omega_\xi}+S_{\check V}$ and a conjugate
  operator $A=\frac{\ii}{2}(v_\xi\cdot\nabla+\nabla\cdot v_\xi)$,
  where $M_{\omega_\xi}$ is the operator of multiplication with the
  real analytic function $\omega_\xi$ which depends real analytically
  on the parameter $\xi$, and the operator $S_{\check V}$ is the
  operator of convolution with the (sufficiently nice) function
  $\check V$, and $v_\xi$ is some vector field determined by
  $\omega_\xi$. Under certain assumptions, which are satisfied for the
  Yukawa potential, we then prove estimates of the form
  $\norm{\ad_A^k(H_V(\xi))(H_0(\xi)+\ii)\inv}\le C_\xi^kk!$ where $C_\xi$
  is some constant which depends continuously on $\xi$. The
  Hamiltonian is the fixed total momentum fiber Hamiltonian of an
  abstract two-body dispersive system and the work is inspired by a
  recent result \cite{EMR} which, under conditions including estimates
  of the mentioned type, opens up for spectral deformation and
  analytic perturbation theory of embedded eigenvalues of finite
  multiplicity.
\end{abstract} {\bfseries Mathematical Subject Classification:} 81Q10,
47B47\newline {\bfseries Keywords:} dispersive systems, iterated
commutators, spectral deformation, Yukawa potentials

\newpage

\section{Introduction}
In this paper we consider a two-body dispersive system where the two
particles interact via a pair-potential $V$. We study the iterated
commutators $\ad_A^k(H_V(\xi))=[\ad_A^{k-1}(H_V(\xi)),A]$ of the fixed
total momentum fiber Hamiltonian $H_V(\xi)$ of this system with an
operator $A$, which (in the sense of Mourre) is conjugate to
$H_V(\xi)$. It is well-known from the literature that the nature of the
spectrum and regularity of eigenstates is related to the (iterated)
commutators with a conjugate operator, see e.g.\
\cite{Mourre,JMP,ABG,HunzikerSigal,CattaneoGrafHunziker,FMS1,FMS2,MW}. See
also \cite{MGR} for another result where control of iterated
commutators is needed. 

Recently, together with Jacob Schach M\o ller, the authors developed
an analytic perturbation theory for embedded eigenvalues in
\cite{EMR}, which also contains an example of a non-trivial model
which satisfies the needed conditions for the abstract theory of that
paper.  We consider a version of this model (introduced in details in
Section~\ref{sec:model}) which has fiber Hamiltonians of the type
\begin{align*}
H_V(\xi)=\omega_{\xi} + S_{\check V},
\end{align*}
where $\omega_\xi$ denotes multiplication by a certain analytic
function (see Condition~\ref{cond_ex}) and $S_{\check V}$ denotes
convolution by the inverse Fourier transform of the interaction
potential $V$.  The abstract conditions in \cite{EMR} involve the
requirement that there exists a constant $C_\xi>0$ such that for all
$k\in\N$, the iterated commutator $\ad_A^k(H_V(\xi))$ exists as a
$H_V(\xi)$-bounded operator and
$\norm{\ad_A^k(H_V(\xi))(H_V(\xi)+\ii)\inv}\le C_\xi^kk!$, which
serves as a motivation to study these iterated commutators.

The result presented in this paper covers the results on the model in
\cite{EMR}, but can deal with more singular potentials such as the
Yukawa potential, see Proposition~\ref{prp:yuk}. The main result of
this paper, Theorem~\ref{cor:totalcommutator}, states the existence of
a constant $C_\xi>0$ such that the bound on the iterated commutators
mentioned above holds, given that $V$ satisfies either
Condition~\ref{cond:Lsw} or Condition~\ref{cond:L1}. In fact, we prove
that this constant is continuous as a function of $\xi$. This allows
us to reach the conclusion of \cite[Theorem~3.2]{EMR}, if certain
further assumptions are satisfied, namely relative boundedness of
$S_{\check V}$ wrt.\ the multiplication operator $\omega_\xi$ and that
$[H_V(\xi),iA]$ satisfies a Mourre estimate. See Theorem~\ref{thm-ex}
for details. We note that these further assumptions have already been
verified for a large class of potentials in \cite{EMR}.
 


The main obstacle in proving the bound on the iterated commutator is
to control the commutator with the interaction term. This is done in
several steps. First, we find a closed formula for the $k$'th iterated
commutator for any $k$, see Theorem~\ref{thm:commutatorformula}. Then
we find appropriate estimates for every term that appears in the
formula. These bounds then turn out to have the right behaviour.

The paper is organized as follows. In Section~\ref{sec:model} the
model is introduced, conditions are stated, and the main results are
formulated. To write the closed commutator formula in a compact,
readable form and for use in proofs and various intermediate results,
we introduce some terminology and notation in
Section~\ref{sec:notation}. In Section~\ref{sec:commutatorformula}, we
then state and prove the main technical result,
Theorem~\ref{thm:commutatorformula}, which is the closed commutator
formula. In Section~\ref{sec:betaell}, a technical lemma is proven
which paves the way for an estimate on the scalar factors in the
formula, which is then stated and proved in
Section~\ref{sec:redordfac}. The next step is to estimate the number
of terms in the sum in the formula, which is done in
Section~\ref{sec:noofterms}. Before stating some sufficient a
posteriori assumptions and making the final estimates on the
interaction commutator, we turn our attention to the iterated
commutator of the free Hamiltonian in Section~\ref{sec:freeham}. This
section is essentially a repetition of known results by the present
authors and J. S. M\o ller, see \cite{EMR}, and is included for the
reader's convenience. The methods used for the free Hamiltonian can be
copied to deal with certain parts of the interaction commutator, which
we return to in Section~\ref{sec:interaction}, where we identify a
posteriori assumptions which are sufficient to prove the right kind of
bounds on the iterated commutators with the interaction term. The main
result in this section is Theorem \ref{thm:interaction} in which we
show that Conditions \ref{cond:Lsw} and \ref{cond:L1} both lead to the
desired bounds. We conclude with Section~\ref{sec:Yukawa} where
Proposition~\ref{prp:yuk}, which states that in dimension $d=3$ Yukawa
potentials satisfy Condition~\ref{cond:Lsw}, is proven.

\section{The model and results}
\label{sec:model}
We introduce a two-particle Hamiltonian on $\mathrm{L}^2(\mdsr^{2d})$ by
\begin{align*}
H_V '=\omega_1(p_1)+\omega_2(p_2) + V(x_1-x_2),
\end{align*}
where $p_i =- \ii\nabla_{x_i}$, $x_i\in\mathds{R}^d$.
 
 We impose the following set of conditions on $\omega_1,\omega_2$ and $V$:

\begin{cond}[Properties of $\omega_1, \omega_2$ and $V$]\label{cond_ex}
\leavevmode
\begin{enumerate}
\item The $\omega_i$'s are real-valued, real analytic functions on $\R^d$ and there exists $R>0$, such that the $\omega_i$'s extend to analytic functions in the $d$-dimensional strip 
\begin{equation*}
S_{2R}^d:=\bigl\{(z_1,\dots, z_d)\in\C^d\,\big|\,|\im (z_i)| <  2R, i=1,\dots, d\bigr\}.
\end{equation*}
We denote the analytic continuations of these functions by the same symbols.
\label{cond_ex_v_omega}
\item\label{item-cond_ex2} There exist real numbers $p=s_2\geq s_1$,
  $s_2>0$ and a constant $C>0$ such that
\begin{equation}\label{eq_omega}
|\partial^\alpha\omega_j(k)|\leq C \langle k\rangle^{s_j} ,\quad|\omega_j(k)|\geq \frac1{C}\langle k\rangle^{s_j} - C
\end{equation}
for every multi-index $\alpha\in\N_0^d$, $|\alpha|\leq 1$ and all $k\in S^d_{2R}$.
\label{cond_ex_omega_est}\label{cond_ex_bnd_on_v}
%
\item The Fourier transform $\hat V$ of $V$ exists and $\hat V\in
  \LL^1_{\textup{loc}}(\R^d)$. 
\end{enumerate}
\end{cond}

Here and hereafter, $\jnorm{x}=\sqrt{x^2+1}$ and $\hat f=f^\wedge$
denotes the Fourier transform of $f$.

Conjugating with the Fourier transform, we see that $H_V'$ is
unitarily equivalent to
\begin{align*}
H_V = \omega_1(k_1) + \omega_2(k_2) + t_V,
\end{align*}
where $t_V$ is the partial convolution operator
\begin{align*}
(t_Vf)(k_1,k_2):= \int_{\R^d} \hat V(u)f(k_1-u,k_2+u)\, \mathrm{d}u
\end{align*}
and
\[
\hat V(k) = (2\pi)^{-d/2} \int_{\R^d} \mathrm{e}^{-\ci k \cdot x} V(x) \, \mathrm{d} x.
\]
In order to fibrate $H_V$ w.r.t. total momentum $\xi = k_1+k_2$, we
introduce a unitary operator
$I \colon \LL^2(\R^d\times\R^d)\to \LL^2(\R^d;L^2(\R^d))$ by setting
\[
(If)(\xi) = f(\xi-\cdot,\cdot).
\]
Under this transformation, we find that the Hamiltonian takes the form
\begin{equation*}
I H_V I^* = \int_{\R^d}^{\oplus} H_V(\xi)\,\mathrm{d} \xi, \quad \textup{where} \quad H(\xi)=H_V(\xi)=\omega_{\xi} + S_{\check V},
\end{equation*}
and
\begin{equation*}
\omega_{\xi}(k)=\omega_1(\xi-k) + \omega_2(k), \qquad (S_{\check V}f)(k)=(\check{V}*f)(k).
\end{equation*}
Here $\check V(k)=\hat V(-k)$ is the inverse Fourier transform of $V$
and $\check V *f$ denotes the convolution product. Note that
$H_0(\xi)$ corresponds to the case where the potential is absent. For
brevity, we will often suppress the subscript $V$ in $H_V(\xi)$ and
just write $H(\xi)$. For later use, let $T_0=S_{\check V}$ for a fixed
$V$.

Furthermore, we define a self-adjoint operator for every total
momentum $\xi\in\R^d$ by
\begin{equation}
A_\xi= \frac{\ci}{2}\bigl( v_{\xi}\cdot \nabla_k + \nabla_k\cdot v_{\xi}\bigr).
\label{eq_A}
\end{equation}
If no confusion can arise, we will just write $A=A_\xi$. The vector
field $v_{\xi}$ is given by
\begin{equation}
v_\xi(k) =  \mathrm{e}^{-k^2-\xi^2}(\nabla_k\omega_\xi)(k).
\label{eq_vectorfield}
\end{equation}
The choice of vector field can be regarded as a generalization of the
standard choice
$A=\frac{\ci}{2} x\cdot\nabla + \frac{\ci}{2}\nabla \cdot x$ in Mourre
theory for Schr\"odinger operators. Note that
$\frac{1}{2}x = \nabla x^2$. Thus, the quadratic dispersion relation
in the Schr\"odinger case gets substituted by a more general one. The
exponential weight in both $x$ and $\xi$ is added to make the uniform
estimates in Sections \ref{sec:freeham} and \ref{sec:interaction}
work.

In addition to Condition~\ref{cond_ex}, which is always assumed, we
will sometimes need one of the following two conditions.  First, we
recall the definition of the weak $\LL^s(\R^d)$ norm:
$\smash{\norm{f}_{s,w} = \sup_{\alpha>0}|\Set{x\given
    \abs{f(x)}>\alpha}|^{\frac{1}{s}}}$.
Here $\abs{A}$ for a subset $A$ of $\R^d$ denotes the Lebesgue measure
of $A$. $\mathrm{L}^s_w(\R^d)$ is then the set of all functions for
which $\norm{f}_{s,w}<\infty$.

\begin{cond}[$\LL^s_w$--bounds on $\hat{V}$]\label{cond:Lsw}
 Let $s>1$ with $\max\{\frac{1}{2},1-\frac{p}{d}\} < \frac{1}{s} \leq \frac{2}{d}$, where $p=s_2$ comes from Condition~\ref{cond_ex}. There exists $c'>0$ such that
 \begin{align*}
 \forall \alpha\in\N_0^d : \|\hat{V}^{(\alpha)}\|_{s,w} \leq \alpha! c'^{|\alpha|},
 \end{align*}  
 where $\norm{\cdot}_{s,w}$ denotes the weak $\LL^s(\R^d)$ norm.	
\end{cond}
\begin{rem}
  Note that $\frac{1}{s}\le\frac{2}{d}$ is only a restriction for
  $d\ge 3$; for lower dimensions the condition $s>1$ implies the former.
\end{rem}
\begin{cond}[$\LL^1$--bounds on $\hat{V}$]\label{cond:L1}
  There exists $c'>0$ such that 
  \begin{align*}
    \forall \alpha\in\N_0^d : \|\hat{V}^{(\alpha)}\|_{1} \leq \alpha! c'^{|\alpha|},
  \end{align*}
\end{cond}

Theorem \ref{thm:interaction} can be proven for a potential satisfying
Condition \ref{cond:Lsw} or \ref{cond:L1}. However it is Condition
\ref{cond:Lsw} that we show to hold in the case of the Yukawa
potential in Proposition \ref{prp:yuk}.

The main result of the paper is now as follows.
\begin{thm}
  \label{cor:totalcommutator}
  Assume Condition~\ref{cond_ex} and \emph{either} Condition~\ref{cond:Lsw}
  \emph{or} Condition~\ref{cond:L1}. Then there exists a constant $C(\xi)$
  which depends continuously on $\xi$, such that
  \begin{equation*}
    \norm{\ad_A^k(H(\xi))(H_0(\xi)+\ii)\inv}\le C(\xi)^kk!\,,
  \end{equation*}
  for all $k\in\N$.
\end{thm}
\begin{proof}
  This follows directly from Proposition~\ref{prp_CommBoundsFree} and
  Theorem~\ref{thm:interaction}.
\end{proof}
\begin{rem}
  In fact, Theorem~\ref{cor:totalcommutator} is true under a weaker
  assumption which can be found in Theorem~\ref{thm:interaction}. The
  logical structure of the argument is that both
  Condition~\ref{cond:Lsw} and Condition~\ref{cond:L1} imply this weaker
  assumption.
\end{rem}
The following Proposition implies that Condition~\ref{cond:Lsw} is
satisfied for the Yukawa potential whenever $p>1$ and $\frac{3}{2}\le
s<2$.
\begin{prp}\label{prp:yuk}
  Let $d=3$, $V(x)= \frac{\mathrm{e}^{-|x|}}{|x|}$ and $s\ge
  3/2$.
  Then $\hat{V}(k)=4\pi(1+k^2)^{-1}$ and there exists $c>0$ such that
  \begin{align*}
    \forall \alpha\in\N_0^3 : \|\hat{V}^{(\alpha)}\|_{s,w} \leq \alpha!c^{|\alpha|}.
  \end{align*}
\end{prp}
\begin{rem}
  Some authors use the Fourier transform $\hat{V}(k)=4\pi(1+k^2)^{-1}$
  to \emph{define} the Yukawa potential in other dimensions. In that
  case, a similar result holds in these dimensions.
\end{rem}
\begin{rem}\label{rem:EMR}
  In \cite{EMR}, the conclusion of Theorem~\ref{cor:totalcommutator}
  is (indirectly) reached by completely different arguments for the
  class of Hamiltonians satisfying the following condition in addition
  to Condition~\ref{cond_ex}:
  \begin{itemize}
  \item\label{cond_ex_V} Let $d' = 2[d/2]+2$. We suppose that $V\in
    C^{d'}(\R^d)$ and there exists $a>0$, such that for all
    $\alpha\in\N_0^d$ with $|\alpha|\leq d'$, we have $\sup_{x\in
      \R^d}\mathrm{e}^{a|x|}|\partial^\alpha_x V(x)|<\infty$.
  \end{itemize}
  It is easy to see that for dispersion relations and potentials
  satisfying these conditions, Condition~\ref{cond:L1} is also
  satisfied. Indeed, one can prove that for
  some $a>0$, $\hat V$ has an analytic continuation to the
  $d$-dimensional strip $S_{a}^d$ and
  \begin{equation*}
    \forall k\in S_{a}^d:\qquad \bigl|\widehat{V}(k)\bigr| \leq C_V \bigl(1+|k|^{d'}\bigr)^{-1},
  \end{equation*}
  see Remark~3.4.4 in \cite{EMR}.  Then, in the case $d=1$,
  \begin{align*}
    \frac{2\pi}{n!}\norm{\hat V^{(n)}}_1&\le \int\Bigl(\int_{\Gamma_r}\frac{\abs{\hat V(w)}}{\abs{w-z}^{n+1}}\dd w\Bigr) \dd z\\
    &\le\int\Bigl(\int_{\tilde\Gamma_r}\frac{ C_V}{1+\abs{r-\abs{z}}^2}\frac{1}{r^{n+1}}\dd w\Bigr)\dd z\\
    &=\frac{2C_V}{r^n}\int\frac{\dd z}{1+\abs{r-\abs{z}}^2},
  \end{align*}
  where $\Gamma_r$ is a path around $z$ with radius $r<a$ and
  $\tilde\Gamma_r$ is a path around $0$. The case $d>1$ is similar.
\end{rem}

In order to apply the abstract spectral deformation theory of
\cite{EMR}, in addition to Theorem~\ref{cor:totalcommutator}, we need
to assume a few extra conditions which are clearly satisfied for the
Yukawa potential and $p>1$, namely relative compactness of the
interaction term and the existence of a Mourre estimate.
\begin{thm}\label{thm-ex} 
  Suppose Condition~\ref{cond_ex} and \emph{either}
  Condition~\ref{cond:Lsw} \emph{or} Condition~\ref{cond:L1}. Assume
  that $T_0=S_{\check V}$ is relatively bounded wrt.\ $H_0$ and that for
  all $\xi\in\R^d$ and all $\lambda\in\R$, there exists positive
  constants $e,C,\kappa>0$ and a compact operator $K$ such that
  \begin{equation*}
    [H(\xi),\ii A_\xi]\ge e-CE_{H(\xi)}(\R\setminus[\lambda-\kappa,\lambda+\kappa])\jnorm{H(\xi)}-K.
  \end{equation*}
  Let $\Sigma_\mathrm{pp}$ be the joint energy-momentum point spectrum
\[
\Sigma_\mathrm{pp} = \bigl\{(\lambda,\xi)\in\R\times\R^d \, \big|\, \lambda\in \Sigma_\mathrm{pp}(\xi)\bigr\},\quad 
\Sigma_\mathrm{pp}(\xi) = \sigma_\mathrm{pp}(H(\xi))
\]
and $\mathcal{T}$ the energy-momentum threshold set
\[
\begin{aligned}
\mathcal T & = \bigl\{ (\lambda,\xi)\in\R\times\R^d \, \big|\, \lambda\in\mathcal T(\xi)\bigr\},\\
\mathcal T(\xi) & = \bigl\{ \lambda\in\R\,\big|\, \exists k\in\R^d: \ \omega_\xi(k) = \lambda \textup{ and } \nabla_k \omega_\xi(k) = 0\bigr\}.
\end{aligned}
\] Let $(\lambda_0,\xi_0)\in\Sigma_{\textup{pp}}\setminus\mathcal{T}$. If we fix
$v\in\R^d$ with $\norm{v}=1$, then there exist
\begin{itemize}
\item $r,\rho>0$
\item natural numbers $0\leq m_\pm\leq n_0$ and $n_1^\pm,\dotsc,n_{m_\pm}^\pm\geq 1$ 
with $n_1^\pm+\cdots +n_{m_\pm}^\pm \leq n_0$,
\item   real analytic functions
$\lambda_1^\pm,\dotsc,\lambda_{m_\pm}^\pm\colon I_\pm \to \R$,  $I_- = (-r,0)$ and $I_+=(0,r)$,
\end{itemize}
such that
\begin{enumerate}
\item for any $j$, $\lim_{t\to0\pm}\lambda_j^\pm(t)=\lambda_0$,
\item for any $t\in I_\pm$, we have
$\sigma_\mathrm{pp}(H(\xi_0+tv))\cap (\lambda_0-\rho,\lambda_0+\rho) = \{\lambda_1^\pm(t),\dotsc,\lambda^\pm_{m_\pm}(t)\}$,
\item The eigenvalue branches
$I_\pm\ni t\to \lambda_j^\pm(t)$ have constant multiplicity $n_j^\pm$.
\end{enumerate}
\end{thm}
This theorem, with the assumptions replaced by just
Condition~\ref{cond_ex} and the condition mentioned in
Remark~\ref{rem:EMR}, is similar to \cite[Theorem~3.2]{EMR}. This
result, however, doesn't cover the Yukawa potential due to the
singularity at $x=0$.
\section{Terminology and notation}
\label{sec:notation}
In this section, we introduce some notation which we use to state and
prove Theorem~\ref{thm:commutatorformula}. The main objects are the
\emph{polyindices} which we usually denote $\alpha$, $\beta$, $a$ or
$b$ (see below), and the two classes of operators indexed by the
polyindices; the multiplication operators $M_{\alpha, \beta}$ and the
convolution operators $T_{\beta, b}$.

In the following, $C_0(\N_0,\N_0^d)$ will denote the set of
non-negative, integer valued functions on $\N_0$ which are zero except
on a finite set. Such functions will be referred to as
\emph{polyindices} of dimension $d$. We will usually use the letters
$\alpha$ and $a$ for $1$-dimensional polyindices, while $\beta$ and
$b$ are reserved for $d$-dimensional polyindices. For any polyindex
$\beta\in C_0(\N_0,\N_0^d)$ of dimension $d$, we will call the finite
number
\begin{equation*}
\norm{\beta}=\sum_{\sigma=1}^d\sum_{i=0}^\infty\beta_\sigma(i)(i+1)
\end{equation*}
the \emph{order} of $\beta$, and the finite number
\begin{equation*}
  \abs{\beta}=\sum_{\sigma=1}^d\sum_{i=0}^\infty\beta_\sigma(i)
\end{equation*}
will be called the \emph{size} of $\beta$. 
The \emph{order factorial} of $\beta$ will be written as, and defined
by,
\begin{equation*}
  \beta\ddot{!}=\prod_{\sigma=1}^d\prod_{i=0}^\infty\beta_\sigma(i)!((i+1)!)^{\beta_\sigma(i)}\,,
\end{equation*}
and is likewise a finite number for any polyindex $\beta$ (all factors
except a finite number are $1$). The \emph{reduced order factorial} of
$\beta$ is then
\begin{equation*}
  \beta\dot{!}=\prod_{\sigma=1}^d\prod_{i=0}^\infty\beta_\sigma(i)!(i+1)^{\beta_\sigma(i)}\,.
\end{equation*}
We will sometimes need the ratio
$\beta\bar!:=\frac{\beta\ddot!}{\beta\dot!}=\prod_{\sigma=1}^d\prod_{i=0}^\infty
(i!)^{\beta_\sigma(i)}$.
\begin{rem}\label{rem_size_of_set}
  At a later point in the paper we will need the following result
  which is easily checked to be true.
  \begin{align*}
    |\{\alpha\in\mathrm{C}_0(\N_0,\N_0)~|~\|\alpha\|=k\}| = p(k),
  \end{align*}
  where $p(k)$ is the number theoretic partition function (see e.g.\
  \cite{HardyWright}) and $|M|$ denotes the number of elements of a
  set $M$. Indeed, if $\norm{\alpha}=k$, then $\alpha$ determines a
  unique partition of $k$ in the following way:
  \begin{equation*}
    k=\norm{\alpha}=\sum_{i=0}^\infty\sum_{j=1}^{\alpha(i)}(i+1),
  \end{equation*}
  and if a partition of $k$ is given, it can be uniquely encoded in an
  $\alpha$ with $\norm{\alpha}=k$ by letting $\alpha(i-1)$ denote the
  number of $i$'s in the partition for all $i\ge1$. A simple upper
  bound for $p(k)$ is $p(k)<e^{\pi\sqrt{2k/3}}=e^{c\sqrt{k}}$, cf.\
  \cite{HardyWright}.
\end{rem}

The polyindices will be used to index certain operators that appear in
the commutator formula for the interaction. More specifically, let
$\alpha$ and $\beta$ be a polyindices of dimension $1$ and $d$,
respectively, and let $f\in C^\infty(\R,\R)$ and
$g\in C^\infty(\R^d,\R)$. Write
\begin{align*}
  \D_vf^\alpha&=\prod_{i=0}^\infty(D_v^if)^{\alpha(i)}&\textup{and}&&\D_vg^\beta&=\prod_{i=0}^\infty\prod_{\sigma=1}^d(D_v^ig_\sigma)^{\beta_\sigma(i)}&\textup{where}&&D_v&=\ii v_\xi\cdot\nabla.
\end{align*}
We define
\begin{equation*}
  M_{\alpha,\beta}:=M_{\D_vw^\alpha\D_vv_\xi^\beta},
\end{equation*}
where $w=\frac{\ii}{2}\div(v_\xi)$ and, for a function $f$, $M_f$
denotes the operator of multiplication with $f$. We note that
$M_{0,0}=\id$ and that $M_{\alpha,\beta}M_{a,b}=M_{\alpha+a,\beta+b}$.

The commutator formula also contains some convolution operators which
are indexed by polyindices. For any pair of $d$-dimensional
polyindices $\beta$ and $b$, we now let the $d$-dimensional multiindex
$\gamma_{\beta+b}$ be defined through
\begin{equation}
  \label{eq:gammabetab}
  (\gamma_{\beta+b})_\sigma=\abs{\beta_\sigma}+\abs{b_\sigma}.
\end{equation}
Then
\begin{equation*}
  T_{\beta, b}:=S_{((-1)^{\abs{b}}(-x)^{\gamma_{\beta+b}} \bar V)^\wedge},
\end{equation*} 
where $S_f$ denotes the operator of convolution with the function $f$
and $\bar V(x)=V(-x)$. Note that, although only $\hat V$ is assumed to
be well-defined, $T_{\beta,b}$ can be interpreted as a form on
$C_0^\infty(\R^d)$ for all $\beta$ and $b$. Observe that
$S_{\hat{\bar{V}}}=T_0=T_{{0,0}}$.

In the proof of the commutator formula, we will make extensive use of
the following notation. If $\gamma\in\N_0^d$ is a multiindex and
$\sigma\in\{1,\dotsc,d\}$, then we let
\begin{equation*}
  \gamma^{+(\sigma)}=\gamma+\delta_\sigma,
\end{equation*}
where $\delta_\sigma$ is given by
$(\delta_\sigma)_{\sigma'}=\delta_{\sigma,\sigma'}$, where
$\delta_{\sigma,\sigma'}$ is the Kronecker delta. For a
$1$-dimensional polyindex $\alpha\in C_0(\N_0,N_0)$, we define
\begin{equation*}
  \alpha^{+(i)}=\alpha+\delta_i,
\end{equation*}
where $\delta_i(j)=\delta_{i,j}$. If $\beta\in
C_0(\N_0,\N_0^d)$ is a $d$-dimensional polyindex, then 
\begin{equation*}
  \beta^{+(i,\sigma)}=\beta+\delta_{(i,\sigma)},
\end{equation*}
where
$(\delta_{(i,\sigma)}(j))_{\sigma'}=\delta_{i,j}\delta_{\sigma,\sigma'}$. Likewise,
we will write
\begin{align*}
  \alpha^{-(i)}&=\alpha-\delta_i
&\textup{ and } &&
  \beta^{-(i,\sigma)}&=\beta-\delta_{(i,\sigma)},
\end{align*}
whenever $\alpha(i),\beta_\sigma(i)\ge1$. If $\alpha(i)$ or
$\beta_\sigma(i)$ is $0$, then $\alpha^{-(i)}$ respectively
$\beta^{-(i,\sigma)}$ can be given any ($1$- respectively
$d$-dimensional poly\-index) value; such cases will only appear in
expressions that are multiplied by $0$.

\section{The commutator formula for the interaction term}
\label{sec:commutatorformula}
\begin{thm}
  \label{thm:commutatorformula}
  Let $T_0=S_{\check V}$ denote the interaction term. Then
  \begin{equation}\label{eq:commutatorformula}
    \ad_A^k(T_0)=\smashoperator{\sum_{\substack{\alpha,\beta,a,b\\\norm{\alpha}+\norm{\beta}+\norm{a}+\norm{b}=k}}}\,\frac{k!}{\alpha\ddot!\beta\ddot!a\ddot!b\ddot!} M_{\alpha,\beta}T_{\beta,b} M_{a,b}
  \end{equation}
  as a form identity on $C_0^\infty(\R^d)$.
\end{thm}
\begin{proof}
  The proof goes by induction. The case $k=0$ is trivially true. Consider 
  \begin{equation*}
    [A,M_{\alpha,\beta}T_{\beta,b}M_{a,b}]=[A,M_{\alpha,\beta}]T_{\beta,b}M_{a,b}+M_{\alpha,\beta}[A,T_{\beta,b}]M_{a,b}+M_{\alpha,\beta}T_{\beta,b}[A,M_{a,b}],
  \end{equation*}
  which here and for the rest of this section should be read as form
  identities on $C_0^\infty(\R^d)$. We thus need to find
  $[A,M_{\alpha,\beta}]$ and $[A,T_{\beta,b}]$:
  \begin{equation*}
    [A,M_{\alpha,\beta}]=\sum_{i=0}^\infty\Bigl(\alpha(i)M_{\alpha^{-(i)+(i+1)},\beta}+\sum_{\sigma=1}^d\beta_\sigma(i)M_{\alpha,\beta^{-(i,\sigma)+(i+1,\sigma)}}\Bigr),
  \end{equation*}
while
\begin{equation*}
  (AT_{\beta,b}f)(k)=(M_{w_\xi}T_{\beta,b}f)(k)+\sum_{\sigma=1}^d(M_{{(v_\xi)}_\sigma}T_{\beta^{+(0,\sigma)},b}f)(k)
\end{equation*}
and
\begin{equation*}
  (T_{\beta,b}Af)(k)=-(T_{\beta,b}M_{w_\xi}f)(k)-\sum_{\sigma=1}^d(T_{\beta,b^{+(0,\sigma)}}M_{{(v_{\xi})}_\sigma}f)(k),
\end{equation*}
so
  \begin{equation*}
    [A,T_{\beta,b}]=M_{\delta_0,0}T_{\beta,b}+T_{\beta,b}M_{\delta_0,0}+\sum_{\sigma=1}^d(M_{0,\delta_{(0,\sigma)}}T_{\beta^{+(0,\sigma)},b}+T_{\beta,b^{+(0,\sigma)}}M_{0,\delta_{(0,\sigma)}}).
  \end{equation*}
  Putting this together (and noting that
  $T_{\beta,b}=T_{\beta^{-(i,\sigma)+(i+1,\sigma)},b}=T_{\beta,b^{-(i,\sigma)+(i+1,\sigma)}}$), we
  get
  \begin{equation}\label{eq:scheme}
    \begin{split}
      \MoveEqLeft[-2] [A,M_{\alpha,\beta}T_{\beta,b}M_{a,b}]=
      \\
      &\sum_{i=0}^\infty\Bigl(\alpha(i)M_{\alpha^{-(i)+(i+1)},\beta}T_{\beta,b}M_{a,b}+\sum_{\sigma=1}^d\beta_\sigma(i)M_{\alpha,\beta^{-(i,\sigma)+(i+1,\sigma)}}T_{{\beta^{-(i,\sigma)+(i+1,\sigma)},b}}M_{a,b}\Bigr)
      \\
      &+M_{\alpha^{+(0)},\beta}T_{\beta,b}M_{a,b}+M_{\alpha,\beta}T_{\beta,b}M_{a^{+(0)},b}
      \\
      &+\sum_{\sigma=1}^d(M_{\alpha,\beta^{+(0,\sigma)}}T_{{\beta^{+(0,\sigma)},b}}M_{a,b}+M_{\alpha,\beta}T_{{\beta,b^{+(0,\sigma)}}}M_{a,b^{+(0,\sigma)}})
      \\
      &+\sum_{i=0}^\infty\Bigl(a(i)M_{\alpha,\beta}T_{\beta,b}M_{a^{-(i)+(i+1)},b}+\sum_{\sigma=1}^db_\sigma(i)M_{\alpha,\beta}T_{{\beta,b^{-(i,\sigma)+(i+1,\sigma)}}}M_{a,b^{-(i,\sigma)+(i+1,\sigma)}}\Bigr).
    \end{split}
  \end{equation}
  From this expression and the induction start, it is clear that the
  $k$'th iterated commutator is of the form
  \begin{equation}\label{eq:inductionstep}
    \ad_A^k(T_0)=\smashoperator{\sum_{\substack{\alpha,\beta,a,b\\\norm{\alpha}+\norm{\beta}+\norm{a}+\norm{b}=k}}}C_{\alpha,\beta,a,b}^{(k)}M_{\alpha,\beta}T_{\beta,b}M_{a,b},
  \end{equation}
  where $\smash{C_{\alpha,\beta,a,b}^{(k)}}$ are some constants to be
  determined. Assume that the commutator formula holds true for
  $k$. We want to show that it also holds for $k+1$. By the above
  discussion, it is enough to let $\alpha,\beta,a,b$ be arbitrary with
  $\norm{\alpha}+\norm{\beta}+\norm{a}+\norm{b}=k+1$ and show that
  $\smash{C_{\alpha,\beta,a,b}^{(k)}}=\frac{k!}{\alpha\ddot!\beta\ddot!a\ddot!b\ddot!}$,
  so this is what we do.
  
  Using the induction hypothesis we combine \eqref{eq:scheme} and
  \eqref{eq:inductionstep} to obtain an expression for
  $\ad_A^{k+1}(T_0)$. This enables us to identify those terms in the
  $k$'th iterated commutator that through commutation with $A$
  contribute to the term
  $\smash{C_{\alpha,\beta,a,b}^{(k+1)}}M_{\alpha,\beta}T_{\beta,b}M_{a,b}$
  in the $k+1$'st iterated commutator. Before proceeding we illustrate
  this by an example. Suppose we are given $\alpha,\beta,a,b$ with
  $\norm{\alpha}+\norm{\beta}+\norm{a}+\norm{b}=k+1$. One of the terms
  appearing in the combination of \eqref{eq:scheme} and
  \eqref{eq:inductionstep} is
  \begin{align*}
  C_{\alpha',\beta',a',b'}^{(k)}\alpha'(i)M_{\alpha'^{-(i)+(i+1)},\beta'}T_{{\beta',b'}}M_{a',b'},
  \end{align*}
  where
  \begin{align*}
 \|\alpha'\|+\|\beta'\|+\|a'\|+\|b'\|=k.
  \end{align*}
  Since the contributing term from the $k$'th commutator can only have
  one polyindex deviating from $\alpha,\beta,a,b$, this term will contribute to $C_{\alpha,\beta,a,b}^{(k+1)}M_{\alpha,\beta}T_{\beta,b}M_{a,b}$, if $\alpha' =  \alpha^{+(i)-(i+1)}$, $\beta'=\beta$, $a'=a$ and $b'=b$. In the same fashion one easily finds
  that all possible contributors have polyindices of one of the
  following forms:
  \begin{align*}
    &(\alpha^{-(0)},\beta,a,b),&&(\alpha^{-(i+1)+(i)},\beta,a,b),&&(\alpha,\beta^{-(0,\sigma)},a,b),&&(\alpha,\beta^{-(i+1,\sigma)+(i,\sigma)},a,b),\\
    &(\alpha,\beta,a^{-(0)},b),&&(\alpha,\beta,a^{-(i+1)+(i)},b),&&(\alpha,\beta,a,b^{-(0,\sigma)}),&&(\alpha,\beta,a,b^{-(i+1,\sigma)+(i,\sigma)}).
  \end{align*}
  Appealing again to the induction hypothesis and \eqref{eq:scheme}, we see that in our example
  \begin{align*}
  \alpha'(i)  C_{\alpha',\beta',a',b'}^{(k)} =   (\alpha+\delta_i-\delta_{i+1})(i)  C_{\alpha+\delta_i-\delta_{i+1},\beta,a,b}^{(k)}=\frac{\alpha(i)(i+2)k!}{\alpha\ddot!\beta\ddot!a\ddot!b\ddot!}.
  \end{align*}
Proceeding in the same way we see that the contributions of all the previously listed terms are, respectively
  \begin{align*}
    &\frac{\alpha(0)k!}{\alpha\ddot!\beta\ddot!a\ddot!b\ddot!},&&
    \frac{\alpha(i+1)(i+2)k!}{\alpha\ddot!\beta\ddot!a\ddot!b\ddot!},&&\frac{\beta_\sigma(0)k!}{\alpha\ddot!\beta\ddot!a\ddot!b\ddot!},&&\frac{\beta_\sigma(i+1)(i+2)k!}{\alpha\ddot!\beta\ddot!a\ddot!b\ddot!},\\
    &\frac{a(0)k!}{\alpha\ddot!\beta\ddot!a\ddot!b\ddot!},&&
    \frac{a(i+1)(i+2)k!}{\alpha\ddot!\beta\ddot!a\ddot!b\ddot!},&&\frac{b_\sigma(0)k!}{\alpha\ddot!\beta\ddot!a\ddot!b\ddot!},&&\frac{b_\sigma(i+1)(i+2)k!}{\alpha\ddot!\beta\ddot!a\ddot!b\ddot!},
  \end{align*}
  times $M_{\alpha,\beta}T_{\beta,b}M_{a,b}$. Finally, summing up all
  these possible contributions gives us the value of
  $C_{\alpha,\beta,a,b}^{(k+1)}$:
  \begin{align*}
  C_{\alpha,\beta,a,b}^{(k+1)} 
  &= 
  \frac{k!}{\alpha\ddot!\beta\ddot!a\ddot!b\ddot!}\sum_{i=0}^\infty (\alpha(i+1)(i+2) + a(i+1)(i+2)) \\
  &~~~+
  \frac{k!}{\alpha\ddot!\beta\ddot!a\ddot!b\ddot!}\sum_{i=0}^\infty\sum_{\sigma=1}^{d}(\beta_\sigma(i+1)(i+2) + b_\sigma(i+1)(i+2)) \\
  &~~~+
  \frac{k!}{\alpha\ddot!\beta\ddot!a\ddot!b\ddot!}\left(\alpha(0) + a(0) + \sum_{\sigma=1}^{d}(\beta_\sigma(0) + b_\sigma(0))\right)\\
  &=
  \frac{k!}{\alpha\ddot!\beta\ddot!a\ddot!b\ddot!}\left(\|\alpha\|+ \|a\|+\|\beta\|+\|b\| \right) = \frac{(k+1)!}{\alpha\ddot!\beta\ddot!a\ddot!b\ddot!},
  \end{align*}
  where we have used that by assumption $\|\alpha\|+ \|a\|+\|\beta\|+\|b\|=k+1$ in the last line. This completes the proof. 
\end{proof}

\section{Estimates on the reduced order factorial of certain polyindices}
\label{sec:betaell}
To be able to use Theorem~\ref{thm:commutatorformula} and
\eqref{eq:commutatorformula} to estimate $\ad_{A}^k(T_0)$, we need
some control on the reduced order factorials, $\alpha\dot!$,
$\beta\dot!$, etc. In fact, we will just use the trivial estimates
$a\dot!,\alpha\dot!\ge1$ for the reduced order factorial of the
$1$-dimensional polyindices pertaining to the multiplication operators
of the form $M_{\D_vw^\alpha}$.

For the $d$-dimensional polyindices $\beta$ and $b$, we will make a
more careful estimate. More precisely, for each
$\sigma\in\{1,2,\dotsc,d\}$, we want to estimate
\begin{equation}
  \label{eq:factorization1}
  C^{\norm{\beta_\sigma}+\norm{b_\sigma}}\beta_\sigma\dot!b_\sigma\dot!=\prod_{i=1}^\infty\biggl(\bigl(\prod_{y=1}^{\beta_\sigma(i)}C^{i+1}y(i+1)\bigr)\bigl(\prod_{y=1}^{b_\sigma(i)}C^{i+1}y(i+1)\bigr)\biggr)
\end{equation}
from below, where $C>0$ is some constant (see the next section for a
``derivation'' of the identity). We will use the factorization
\eqref{eq:factorization1} and split the estimate into two parts. The
first part is a careful estimate on the product of the smallest factors
in \eqref{eq:factorization1}, the last will be a rough estimate on the
rest of the factors.

The first part part is contained in the following technical lemma, the
rest is postponed to the next section.

\begin{lem}
  \label{lem:betaell}
  For any $C>0$ and $\ell\in\N$, let $\beta^\ell\colon\N_0\to\N_0$
  denote the polyindex given by
  \begin{equation*}
    \beta^\ell(i)=\biggl\lfloor\frac{\ell}{C^i(i+1)}\biggr\rfloor,
  \end{equation*}
  where $\lfloor x\rfloor$ denotes the integer part of $x$.  For any
  $0<\epsilon<1$, there exists a $C_0$ such that for all $C>C_0$ and
  all $\ell$,
  \begin{equation}
    \label{eq:betaellinequality2}
    \abs{\beta^\ell}\le(1+\epsilon)\ell
  \end{equation}
  and
  \begin{equation}
    \label{eq:betaellinequality1}
    C^{2\norm{\beta^\ell}}(\beta^\ell\dot!)^2\ge C''^{2\abs{\beta^\ell}}(2\abs{\beta^\ell})!,
  \end{equation}
  where $C''=C/(4e(1+\epsilon))$.
\end{lem}
\begin{proof}
  The size $\abs{\beta^\ell}$ of $\beta^\ell$ clearly depends on $C$
  and can be estimated from above (for sufficiently large $C\ge C_0$)
  in the following way:
  \begin{align*}
    1\le\ell\le \abs{\beta^\ell}&=\sum_{i=0}^\infty\biggl\lfloor\frac{\ell}{C^i(i+1)}\biggr\rfloor=\ell+\sum_{i=1}^\infty\biggl\lfloor\frac{\ell}{C^i(i+1)}\biggr\rfloor\le\ell+\int_0^\infty\frac{\ell}{C^x(x+1)}\dd x\\
    &=\ell+\bigl[\ell C\Ei(-(x+1)\log(C))\bigr]_{x=0}^\infty\\
    &=\ell-\ell C\Ei(-\log(C))\,,
  \end{align*}
  where $\Ei$ denotes the exponential integral function, see e.g.\
  \cite{NIST} for the definition and properties. Since
  \begin{equation*}
    x\Ei(-\log(x))\sim-\frac{1}{\log(x)}\quad\textup{ as } x\to\infty,
  \end{equation*}
  for any $\epsilon>0$, we can pick $C_0$ so that for $C>C_0$, we have
  $-C\Ei(-\log(C))<\epsilon$. Putting this together, we get
  \begin{equation*}
    \abs{\beta^\ell}\le(1+\epsilon)\ell
  \end{equation*}
  for any $\ell$ and $C>C_0$ which is \eqref{eq:betaellinequality2}.

  Let $W$ denote the Lambert $W$-function (again, see
  \cite{NIST}). Then, using Stirling's formula, we can estimate\allowdisplaybreaks
  \begin{align}
    \log(C^{\norm{\beta^\ell}}\beta^\ell\dot!)
    &=\sum_{i=0}^\infty\log\Bigl(C^{\bigl\lfloor\frac{\ell}{C^i(i+1)}\bigr\rfloor (i+1)}\bigl\lfloor\tfrac{\ell}{C^i(i+1)}\bigr\rfloor!(i+1)^{\bigl\lfloor\frac{\ell}{C^i(i+1)}\bigr\rfloor}\Bigr)\nonumber\\
    &=\smashoperator{\sum_{i=0}^{\bigl\lfloor\frac{W(\ell C\log(C))}{\log(C)}-1\bigr\rfloor}}{\bigl\lfloor\tfrac{\ell}{C^i(i+1)}\bigr\rfloor (i+1)}\log(C)+\log\bigl(\bigl\lfloor\tfrac{\ell}{C^i(i+1)}\bigr\rfloor!\bigr)+{\bigl\lfloor\tfrac{\ell}{C^i(i+1)}\bigr\rfloor}\log(i+1)\nonumber\\
    &\ge\smashoperator{\sum_{i=0}^{\bigl\lfloor\frac{W(\ell C\log(C))}{\log(C)}-1\bigr\rfloor}}{\bigl\lfloor\tfrac{\ell}{C^i(i+1)}\bigr\rfloor}\bigl((i+1)\log(C)+\log\bigl(\bigl\lfloor\tfrac{\ell}{C^i(i+1)}\bigr\rfloor\bigr)-1+\log(i+1)\bigr)\nonumber\\
    &\ge\smashoperator{\sum_{i=0}^{\bigl\lfloor\frac{W(\ell
          C\log(C))}{\log(C)}-1\bigr\rfloor}}{\bigl\lfloor\tfrac{\ell}{C^i(i+1)}\bigr\rfloor}\bigl(\log\bigl(C^{i+1}(i+1)\bigl\lfloor\tfrac{\ell}{C^i(i+1)}\bigr\rfloor\bigr)-1\bigr)\,.\label{eq:orderfactorialestimate1}
  \end{align}
  Since $\bigl\lfloor\tfrac{\ell}{C^i(i+1)}\bigr\rfloor=k$ for some
  $k\ge1$ 
  when $i\in\bigl\{0,\dotsc,\bigl\lfloor\frac{W(\ell
    C\log(C))}{\log(C)}-1\bigr\rfloor\bigr\}$ and hence
  $\frac{\ell}{C^i(i+1)}<k+1$, we see that
  $\frac{C\ell}{k+1}<C^{i+1}(i+1)$ and hence, using
  \eqref{eq:betaellinequality2}, we can estimate
  $C^{i+1}(i+1)\bigl\lfloor\tfrac{\ell}{C^i(i+1)}\bigr\rfloor>\frac{C\ell}{k+1}k\ge\frac{C\ell}{2}=C'2(1+\epsilon)\ell\ge
  C'2\abs{\beta^\ell}$ where $C'=C/(4(1+\epsilon))$. Then
  \eqref{eq:orderfactorialestimate1} can be estimated from below by
  \begin{equation*}
    \eqref{eq:orderfactorialestimate1}\ge\smashoperator{\sum_{i=0}^{\bigl\lfloor\frac{W(\ell C\log(C))}{\log(C)}-1\bigr\rfloor}}\bigl\lfloor\tfrac{\ell}{C^i(i+1)}\bigr\rfloor\bigl(\log(2\abs{\beta^\ell})+\log(C')-1\bigr)=\abs{\beta^\ell}\bigl(\log(2\abs{\beta^\ell})+\log(C')-1\bigr).
  \end{equation*}
  This means that, again appealing to Stirling's formula,
  \begin{equation*}
    C^{2\norm{\beta^\ell}}(\beta^\ell\dot!)^2\ge\biggl(\frac{2\abs{\beta^\ell}C'}e\biggr)^{2\abs{\beta^\ell}}\ge C''^{2\abs{\beta^\ell}}(2\abs{\beta^\ell})!\,,
  \end{equation*}
  where $C'' = C'/e=C/(4e(1+\epsilon))$, which is \eqref{eq:betaellinequality1}.
\end{proof}
\section{Estimates on the reduced order factorial of general polyindices}
\label{sec:redordfac}
In this section, we prove a result, Corollary~\ref{cor:redordfac},
which we need to control the order factorials in the commutator
formula \eqref{eq:commutatorformula} from
Theorem~\ref{thm:commutatorformula}. In Section~\ref{sec:betaell}, we
took care of a special case in Lemma~\ref{lem:betaell}. Here, we show
that we can split the general case into two factors, one which is
handled by Lemma~\ref{lem:betaell}, and another, which can be handled
by a simple estimate.
\begin{lem}
  \label{lem:redordfac}
  There exists a constant $C_0>0$ such that for all $C>C_0$ and all
  polyindices $\beta, b$,
  \begin{equation*}
    c^{\abs{\gamma_{\beta+b}}}\gamma_{\beta+b}!\le
    C^{\norm{\beta}+\norm{b}}\beta\dot!b\dot!,
  \end{equation*}
  where $c$ can be chosen as $C/(8e)$.
\end{lem}
\begin{rem}
  \label{rem:redordfac}
  Note that the constant $c$ can be chosen arbitrarily large as long
  as $C$ is adjusted accordingly. Note also that
  Lemma~\ref{lem:redordfac} is ``sharp'' in the sense that for all $k$
  and all multiindices $\gamma$ with $\abs{\gamma}=k$, and $\beta$
  given by $\beta(0)=\gamma$ and $\beta_\sigma(i)=0$ for $i\ge1$ and
  all $\sigma$, one has that $\norm{\beta}=k$ and
  $\gamma!=\gamma_\beta!=\beta\dot!$.
\end{rem}
\begin{proof}

  We begin by observing that since we can factorize according to
  dimension
  $c^{\abs{\gamma_{\beta+b}}}\gamma_{\beta+b}!=\prod_{\sigma=1}^dc^{\abs{\beta_\sigma}+\abs{b_\sigma}}(\abs{\beta_\sigma}+\abs{b_\sigma})!$
  and
  $C^{\norm{\beta}+\norm{b}}\beta\dot!b\dot!=\prod_{\sigma=1}^dC^{\norm{\beta_\sigma}+\norm{b_\sigma}}\beta_\sigma\dot!b_\sigma\dot!$,
  it is enough to prove that for any pair $\beta,b$ and any $\sigma$,
  $1\le\sigma\le d$, we have
  \begin{equation}
    c^{\abs{\beta_\sigma}+\abs{b_\sigma}}(\abs{\beta_\sigma}+\abs{b_\sigma})!\label{eq:factorialinequality1}\le C^{\norm{\beta_\sigma}+\norm{b_\sigma}}\beta_\sigma\dot!b_\sigma\dot!,
  \end{equation}
  for some constants $c$ and $C$.  Let $\beta,b$ and $\sigma$ be
  given, let $0<\epsilon<1$ and $C>C_0$ with $C_0$ as in
  Lemma~\ref{lem:betaell}, and let $\ell\in\N$ be the largest number
  such that
  $2\abs{\beta^\ell}\le\abs{\beta_\sigma}+\abs{b_\sigma}$. Rewrite
  \begin{align}
    C^{\norm{\beta_\sigma}+\norm{b_\sigma}}\beta_\sigma\dot!b_\sigma\dot!&=\prod_{i=0}^\infty C^{(\beta_\sigma(i)+b_\sigma(i))(i+1)}\beta_\sigma(i)!(i+1)^{\beta_\sigma(i)}b_\sigma(i)!(i+1)^{\beta_\sigma(i)}
    \nonumber
    \\
    &=\prod_{i=0}^\infty\biggl(\prod_{y=1}^{\beta_\sigma(i)}
    C^{(i+1)}y(i+1)\biggr)\biggl(\prod_{y=1}^{b_\sigma(i)}
    C^{(i+1)}y(i+1)\biggr)=\smashoperator{\prod_{j=1}^{\abs{\beta_\sigma}+\abs{b_\sigma}}}p_j,
    \label{eq:p_j}
  \end{align}
  for some $p_j\le p_{j+1}$, i.e.\ a product of
  $\abs{\beta_\sigma}+\abs{b_\sigma}$ factors of the form $p_j=C^{i+1}y(i+1)$
  with either $y\le\beta_\sigma(i)$ or $y\le b_\sigma(i)$. Replacing
  $\beta_\sigma$ and $b_\sigma$ by $\beta^\ell$ in the above identity
  yields
  \begin{equation}\label{eq:factorization3}
    C^{2\norm{\beta^\ell}}(\beta^\ell\dot!)^2=\prod_{i=0}^\infty\biggl(\smashoperator[r]{\prod_{y=1}^{\bigl\lfloor\frac{\ell}{C^i(i+1)}\bigr\rfloor}}C^{i+1}y(i+1)\biggr)^{\!\!2},
  \end{equation}
  where, for each $i$,
  $y$
  runs through exactly those integers for which $C^{i+1}y(i+1)\le
  C\ell$. Note that \eqref{eq:factorization3} contains exactly
  $2\abs{\beta^\ell}$
  (non-trivial) factors. We will now compare the first
  $2|\beta^\ell|$
  factors in \eqref{eq:p_j} with the factors appearing in
  \eqref{eq:factorization3}. More precisely, we split the ordered
  product in \eqref{eq:p_j} into those which also appear in
  \eqref{eq:factorization3} and a remainder. Note that by the above
  discussion this splitting corresponds to sorting the relevant
  $p_j$
  into those less than or equal to $C\ell$
  and those strictly larger than $C\ell$.
  The first group can be written as
  \begin{equation*}
    \prod_{p_j\le C\ell}p_j=\prod_{i=0}^\infty\biggl(\smashoperator[r]{\prod_{y=1}^{\beta_\sigma(i)\land\beta^\ell(i)}}
    C^{(i+1)}y(i+1)\biggr)\biggl(\smashoperator[r]{\prod_{y=1}^{b_\sigma(i)\land\beta^\ell(i)}}
    C^{(i+1)}y(i+1)\biggr),
  \end{equation*}
  where $f\land g$ denotes the minimum of $f$ and $g$ and the whole product can then be rewritten as in \eqref{eq:factorization2}. 
  \begin{align}
    \prod_{j=1}^{2\abs{\beta^\ell}}p_j&=\Biggl(\prod_{i=0}^\infty\biggl(\smashoperator[r]{\prod_{y=1}^{\beta_\sigma(i)\land\beta^\ell(i)}}
    C^{(i+1)}y(i+1)\biggr)\biggl(\smashoperator[r]{\prod_{y=1}^{b_\sigma(i)\land\beta^\ell(i)}}
    C^{(i+1)}y(i+1)\biggr)\Biggr)\smashoperator{\prod_{\substack{j=1\\p_j>C\ell}}^{2\abs{\beta^\ell}}}p_j
    \label{eq:factorization2}
    \\
    &\ge C''^{2\abs{\beta^\ell}}(2\abs{\beta^\ell})!.
    \label{eq:suffpotproof1}
  \end{align}
  To obtain the estimate \eqref{eq:suffpotproof1} simply note that for
  each term missing to obtain
  $C^{2\|\beta^\ell\|}{\beta^\ell\dot{!}}$
  in the first product we find one term in the remainder for which
  $p_j>C\ell$.
  Since the missing term's value must be less than or equal to
  $C\ell$,
  we may estimate this $p_j$
  from below by this missing value. The claimed inequality then
  follows from \eqref{eq:betaellinequality1} in
  Lemma~\ref{lem:betaell}.

  Since $\ell$ was chosen as the largest integer such that
  $2\abs{\beta^\ell}\le\abs{\beta_\sigma}+\abs{b_\sigma}$, we have
  that 
  $\abs{\beta_\sigma}+\abs{b_\sigma}<2\abs{\beta^{\ell+1}}\le2(1+\epsilon)(\ell+1)$,
  where we for the last inequality used \eqref{eq:betaellinequality2}
  from Lemma~\ref{lem:betaell}.  This implies that
  $C''(\abs{\beta_\sigma}+\abs{b_\sigma})<C\ell$
  . For the remaining $p_j$, we then note that
  \begin{equation}\label{eq:suffpotproof2}
    \smashoperator{\prod_{j=2\abs{\beta^\ell}+1}^{\abs{\beta_\sigma}+\abs{b_\sigma}}}p_j\ge (C\ell)^{\abs{\beta_\sigma}+\abs{b_\sigma}-2\abs{\beta^\ell}}>C''^{\abs{\beta_\sigma}+\abs{b_\sigma}-2\abs{\beta^\ell}}(\abs{\beta_\sigma}+\abs{b_\sigma})_{\abs{\beta_\sigma}+\abs{b_\sigma}-2\abs{\beta^\ell}}\,,
  \end{equation}
  where $(x)_n=x(x-1)\dotsb(x-(n-1))$ denotes the falling
  factorial. Putting \eqref{eq:suffpotproof1} and
  \eqref{eq:suffpotproof2} together yields
  \eqref{eq:factorialinequality1} with $c=C''$.
\end{proof}
\begin{cor}
  \label{cor:redordfac}
  With $c$ and $C$ as in Lemma~\ref{lem:redordfac}, we have the following estimate:
  \begin{equation*}
    \frac{k!}{\alpha\ddot!\beta\ddot!a\ddot!b\ddot!}\le\frac{C^{\norm{\beta}+\norm{b}}k!}{c^{\abs{\gamma_{\beta+b}}}\gamma_{\beta+b}!\beta\bar!b\bar!}\,.
  \end{equation*}
\end{cor}
\begin{proof}
  This follows easily from Lemma~\ref{lem:redordfac}.
\end{proof}
\section{The size of the summation index set}
\label{sec:noofterms}
In this section, we show that the number of terms in the commutator
formula \eqref{eq:commutatorformula} from
Theorem~\ref{thm:commutatorformula} grows in a controllable way.
\begin{prp}
  \label{prp:noofterms}
  The number of terms in the iterated commutator formula for the
  interaction term \eqref{eq:commutatorformula} from
  Theorem~\ref{thm:commutatorformula} is bounded by $c_d^k$, where
  $c_d$ is some constant which only depends on the dimension $d$.
\end{prp}
\begin{proof}
  For any fixed $k$, the set
  \begin{equation*}
    \{(\alpha,\beta,a,b):\alpha,a\in
    C_0(\N_0,\N_0), \beta,b\in
    C_0(\N_0,\N_0^d),\norm{\alpha}+\norm{\beta}+\norm{a}+\norm{b}=k\},
  \end{equation*}
  where $C_0(\N_0,\N_0)$ and $C_0(\N_0,\N_0^d)$ denote the sets of
  $1$- and $d$-dimensional polyindices, respectively, is the index set
  for the summation formula \eqref{eq:commutatorformula} from
  Theorem~\ref{thm:commutatorformula} for the $k$'th iterated
  commutator of the interaction term $T_0$ with the conjugate operator
  $A$.  It can also be written as
  \begin{align*}
    \{(\alpha,\beta,a,b):{}&\alpha,a,\beta_\sigma,b_\sigma\in
    C_0(\N_0,\N_0),
    \beta=(\beta_\sigma)_{\sigma=1}^d,b=(b_\sigma)_{\sigma=1}^d,\\&\norm{\alpha}+\norm{a}+\sum_{\sigma=1}^d\norm{\beta_\sigma}+\norm{b_\sigma}=k\}.
  \end{align*}
  For any weak composition (see e.g.\ \cite{HMComb})
  $\sum_{j=1}^{2d+2}k'_j=k$ of $k$ with exactly $2d+2$ terms, there
  are
  \begin{equation*}
    \prod_{j=1}^{2d+2}\#\{\alpha\in C_0(\N_0,\N_0):\norm{\alpha}=k'_j\}
  \end{equation*}
  different ways of satisfying the condition:
  \begin{align*}
    \norm{\alpha}&=k'_1,&\norm{a}&=k'_2,&
    \norm{\beta_\sigma}&=k'_{\sigma+2},&\norm{b_\sigma}&=k'_{\sigma+d+2},\quad\textup{for
    }\sigma=1,\dotsc,d.
  \end{align*}
  By Remark \ref{rem_size_of_set} the set $\{\alpha\in C_0(\N_0,\N_0)\colon\norm{\alpha}=k\}$ has
  exactly $p(k)$ elements. 

  We now want to rewrite the weak composition
  $\sum_{j=1}^{2d+2}k'_j=k$ of $k$ with exactly $2d+2$ terms in the
  following way. Let $k_j=\sum_{n=1}^jk'_n$ for
  $j=1,\dotsc,2d+1$. Then $k'_1=k_1$,
  \begin{equation*}
    k'_j=k_j-k_{j-1}\quad\textup{for}\quad 2\le j\le2d+1,
  \end{equation*}
  $k'_{2d+2}=k-k_{2d+1}$, and $0\le k_j\le k_{j+1}\le k$ for any
  $j=1,\dotsc,2d$. This means that any weak $2d+2$-term composition of
  $k$ is given by a finite, increasing sequence
  $\{k_j\}_{j=1}^{2d+1}$, i.e.\ satisfying $0\le k_1\le
  k_2\le\dotsb\le k_{2d}\le k_{2d+1}\le k$. We can now count -- and
  estimate -- the number of weak $2d+2$-term compositions of $k$ by
  counting the number of ways we can choose the sequence
  $\{k_j\}_{j=1}^{2d+1}$:
  \begin{align}
  	&    \smashoperator{\sum_{\substack{k_1,k_2,\dotsc,k_{2d},k_{2d+1}\\0\le k_1\le k_2\le\dotsb\le k_{2d}\le k_{2d+1}\le k}}}1 ~~~~~~~~~~=
    \binom{k+2d+1}{2d+1} \nonumber\\
    &= 
    \frac{(k+1)(k+2)\dotsb(k+2d)(k+2d+1)}{(2d+1)!}<C_d^k,\label{eq:binominequality}
  \end{align}
  for some sufficiently large constant $C_d>0$. Note that in the first step of \eqref{eq:binominequality} we have used that the sequence of positive integers $k_1 \leq k_2 \leq \dots \leq k_{2d+1}$ with $0\leq k_i\leq k$ can be identified with a strictly increasing sequence $h_1 < h_2 < \dots <h_{2d+1}$ in a 1 to 1 fashion by putting $h_i=k_i+i$. However, the number of possible choices for such $h_i$ obeying $1\leq h_i \leq k+2d+1$ is exactly $k+2d+1$ choose $2d+1$ as claimed in the first step of \eqref{eq:binominequality}.

  We can now estimate the number of elements in the summation index
  set from above in the following way:
  \begin{align*}
    \MoveEqLeft\# \{(\alpha,\beta,a,b):\alpha,a\in C_0(\N_0,\N_0),
    \beta,b\in
    C_0(\N_0,\N_0^d),\norm{\alpha}+\norm{\beta}+\norm{a}+\norm{b}=k\},
    \\
    &=\smashoperator{\sum_{0\le k_1\le k_2\le\dotsb\le k_{2d}\le
        k_{2d+1}\le k}}p(k_1)p(k_2-k_1)\dotsb
    p(k_{2d+1}-p_{2d})p(k-k_{2d+1})
    \\
    &<\smashoperator{\sum_{0\le k_1\le k_2\le\dotsb\le k_{2d}\le
        k_{2d+1}\le
        k}}e^{c(\sqrt{k_1}+\sum_{j=2}^{2d+1}\sqrt{k_j-k_{j-1}}+\sqrt{k-k_{2d+1}})}
    \\
    &\le\smashoperator{\sum_{0\le k_1\le k_2\le\dotsb\le k_{2d}\le
        k_{2d+1}\le k}}e^{ck}
    \\
    &<(C_de^c)^k,
  \end{align*}
  where we used \eqref{eq:binominequality} for the last
  inequality. 
\end{proof}

\section{Commutators with the free Hamiltonian}
\label{sec:freeham}
In this section we assume that $V=0$ so that $H(\xi)=H_0(\xi)$ is
simply given by multiplication with the function $\omega_{\xi}$. Since
$A=\frac{\ii}{2}\mathrm{div}(v_\xi) + \ii v_{\xi}\cdot\nabla_k$, it is
easy to see that the commutator form $[A,M_f]$, where $f$ is some
function, is given by the operator
\begin{align*}
\ad_{A}(M_f) &= M_{\ii v_{\xi}\cdot \nabla_k f}.
\end{align*} 
If the gradient is finite almost everywhere, this operator is again
just a multiplication operator by a bounded function and is thus
bounded as well. In particular, this is true for the choice
$f=\omega_\xi$, see Section~\ref{sec:model}. Furthermore, we may
iterate the preceeding calculation and obtain that the $n$'th
commutator form is given by the bounded multiplication operator
\begin{align}\label{eq:it_comm_f}
\ad^n_{A}(M_f) &= M_{( \ii v_{\xi}\cdot \nabla_k)^n f}, ~~~n\geq 1
\end{align}
provided the $n$-th derivatives of f remain finite. As noted above the choice $f=\omega_\xi$ thus yields a bounded operator.

\begin{prp}[Commutator Bounds in the Free Case]	\label{prp_CommBoundsFree}
  For all $n\in\mathds{N}$ the iterated commutator
  $\ad_{A}^n(H_0(\xi))$ is given by a bounded multiplication
  operator as follows
  \begin{equation}\label{eq:it_comm_free}
    \ad^n_{A}(H_0(\xi)) = M_{( \ii v_{\xi}\cdot \nabla_k)^n\omega_{\xi}},
  \end{equation}
  where $M_f$ is the operator given by multiplication with $f$, and
  there exists a constant $C_\xi>0$ independent of $n$ and $h$, which
  depends continuously on $\xi\in\R^d$ such that we have the pointwise
  estimate
  \begin{equation}\label{eq:pw_est_omega}
    \abs{(\ii v_{\xi}\cdot \nabla_k)^n\omega_{\xi}(h)} \leq C_\xi^n n! \jnorm{h}^{2s_2}e^{-h^2}
  \end{equation}
  for all $h\in\R^d$. In particular, there exists a constant $c_\xi$
  which depends continuously on $\xi$ such that for all $k\in\N$,
  \begin{equation}\label{eq:est_it_comm_free}
    \norm{\ad_A^k(H_0(\xi))}\le c_\xi^kk!\,.
  \end{equation}
\end{prp}
\begin{proof}
  \eqref{eq:it_comm_free} follows directly from \eqref{eq:it_comm_f} with the choice $f=\omega_\xi$ and \eqref{eq:est_it_comm_free} is implied by \eqref{eq:pw_est_omega}. It thus suffices to prove \eqref{eq:pw_est_omega}. Note that
  \begin{align*}
    (v_{\xi}\cdot \ii\nabla_k )^n \omega_{\xi} (k) &=
    \frac{\mathrm{d}^{n-1}}{\mathrm{d}^{n-1}s}u_{\xi}^s(k)|_{s=0},
    ~~~~u_{\xi}^s(k):=(v_{\xi}\cdot\ii\nabla_k
    \omega_{\xi})(\gamma_s(k)),
  \end{align*}
  where $\gamma_s$ solves the ODE
  \begin{align*}
    \frac{\mathrm{d}}{\mathrm{d}s}\gamma_s^\xi(k)=v_{\xi}(\gamma_s^\xi(k)),
    \qquad\gamma_0(k)=k.
  \end{align*}
  By Lemma~3.5 of \cite{EMR}, for all $k,\xi\in\R^d$ the map
  $s\mapsto\gamma_s^\xi(k)$ extends analytically to a strip of some width
  $r>0$ independent of $k$ and $\xi$, such that $S_r\ni z\mapsto
  \gamma_z^\xi(k)\in S_R^d$. Moreover, there exists a constant $C_\omega>0$, which
  is also independent of $k,\xi\in\R^d$ such that
  $\abs{\gamma_z^\xi(k)-k}\le C_\omega\abs{z}$, see the Remark~3.8 of \cite{EMR}.
  Thus, we may use Cauchy's integral formula to calculate
  \begin{align*}
    |(v_{\xi}\cdot \ii\nabla_k )^n \omega_{\xi} (h)| &=
    \left|\frac{\mathrm{d}^{n-1}}{\mathrm{d}^{n-1}s}u_{\xi}^s(h)|_{s=0}\right| = \left|\frac{(n-1)!}{2\pi\ii}\int\limits_{\Gamma_{r}}\frac{u_{\xi}(\gamma_z^\xi(h))}{z^n}\mathrm{d}z\right|\\
    &\leq (n-1)!\frac{\sup_{z\in
        \Gamma_r}|u_{\xi}(\gamma_z^\xi(h))|}{r^{n-1}},
  \end{align*}
  where $\Gamma_r$ denotes the set $\{\abs{z}=r\}$, and, by abuse of
  notation, a path parametrizing this set in the counter-clockwise
  direction.  By applying Peetre's inequality,
  \begin{equation*}
    \forall q\in\R, k,h\in\C^d\colon\quad\jnorm{k+h}^q\le2^{\abs{q}}\jnorm{k}^{\abs{q}}\jnorm{h}^q,
  \end{equation*}
  see e.g.\ \cite[Lemma~1.18]{X}, the assumptions, and the estimate
  $\abs{\gamma_z^\xi(k)-k}\le C_\omega\abs{z}$, we may estimate
  \begin{equation*}
    \sup_{z\in\Gamma_r}\abs{u_\xi(\gamma_z^\xi(h))}\le \tilde C_\xi\jnorm{h}^{2s_2}e^{-h^2},
  \end{equation*}
  for some constant $\tilde C_\xi\ge1$, which can be chosen such that
  it depends continuously on $\xi\in\R^d$. Since $\tilde C_\xi\ge1$,
  we may define $C_\xi=\frac{\tilde C_\xi}{\max\{1,r\}}$ and conclude the
  statement.
\end{proof}

\section{Estimates on the interaction commutator}
\label{sec:interaction}
In this section, we make estimates on the commutator from
Theorem~\ref{thm:commutatorformula}, using the estimates established
in Sections~\ref{sec:betaell} and \ref{sec:redordfac} on the order
factorials, and in Section~\ref{sec:noofterms} on the number of
terms. We also use Proposition~\ref{prp_CommBoundsFree} from
Section~\ref{sec:freeham} to control the multiplication operators
$M_{\alpha, \beta}$.
\begin{lem}
  \label{lem:v_xi}
  There exists a constant $C_\xi'$ which depends continuously on
  $\xi\in\R^d$ such that for all $h\in\R^d$, all $n\in\N$, and all
  $\sigma\in\{1,2,\dotsc,d\}$, we have the pointwise estimate
  \begin{equation*}
    \abs{(\ii v_\xi\cdot\nabla)^n v_{\xi,\sigma}(h)}\le C_\xi'^nn!\jnorm{h}^{2s_2}e^{-h^2}.
  \end{equation*}
\end{lem}
\begin{proof}
  Mimic the proof of Proposition~\ref{prp_CommBoundsFree} with
  $\omega_\xi$ replaced by $v_{\xi,\sigma}$.
\end{proof}
\begin{prp}[Estimates on $M_{\alpha,\beta}$]
  \label{prp:estimatesMab}
  Let $\alpha$ and $\beta$
  be $1$- and $d$-dimensional polyindices, respectively, and write
  $f_{\alpha,\beta}=\mathbb{D}_v w^{\alpha}\mathbb{D}_v
  v_{\xi}^{\beta}$, such that $M_{\alpha,\beta}=M_{f_{\alpha,\beta}}$. Then we have the following pointwise estimate
  \begin{equation*}
    \abs{f_{\alpha,\beta}(h)}\le C_\xi''^{\norm{\alpha}+\norm{\beta}}(\jnorm{h}^{2s_2}e^{-h^2})^{\abs{\alpha}+\abs{\beta}}\alpha\bar!\beta\bar!\,,
  \end{equation*}
  where $C_\xi''=\max\{C_\xi,C_\xi'\}$ and $C_\xi$ and $C_\xi'$ are
  the constants from Proposition~\ref{prp_CommBoundsFree} and
  Lemma~\ref{lem:v_xi}, respectively.
\end{prp}
\begin{proof}
  Note that $M_{\alpha,\beta}=M_{\mathbb{D}_v
    w^{\alpha}}M_{\mathbb{D}_v v_{\xi}^{\beta}}$, and that we may
  write
  \begin{align*}
    M_{\mathbb{D}_v w^{\alpha}} = \prod_{i=0}^{\infty}M_{(D_v^i
      w)^{\alpha(i)}}. 
  \end{align*}
  Put $w_\xi^0(k):=w(k)$ and note that $((\ci v_\xi\cdot\nabla_k)^n\omega_\xi)(k) = ((\ci v_\xi\cdot\nabla_k)^{n-1}w_\xi)(k)$. We now use Proposition~\ref{prp_CommBoundsFree} to get the
  pointwise estimate
  \begin{align*}
    \abs{\mathbb{D}_v w^{\alpha}(h)} &=
    \prod_{i=0}^{\infty}\abs{((D_v^i w)(h))^{\alpha(i)}} \leq \prod_{i=0}^{\infty} (C_\xi^i
    i!\jnorm{h}^{2s_2}e^{-h^2})^{\alpha(i)}
    \\
    &\leq C_\xi^{\|\alpha\|}(\jnorm{h}^{2s_2}e^{-h^2})^{\abs{\alpha}}
    \prod_{i=0}^{\infty}i!^{\alpha(i)}.
  \end{align*}
  Likewise, we note that
  \begin{equation*}
    M_{\mathbb{D}_v v_{\xi}^{\beta} }= \prod_{\sigma=1}^d
    \prod_{i=0}^{\infty} M_{(D_v^iv_{\xi,\sigma})^{\beta_{\sigma}(i)}}
  \end{equation*}
  which we use to compute the pointwise estimate as before
  \begin{equation*}
    \abs{M_{\mathbb{D}_v v_{\xi}^{\beta} }(h)} \le C_\xi'^{\norm{\beta}}(\jnorm{h}^{2s_2}e^{-h^2})^{\abs{\beta}}\prod_{\sigma=1}^d\prod_{i=1}^\infty i!^{\beta_\sigma(i)}.
  \end{equation*}
  Combining these two estimates now gives the proposition.
\end{proof}
\begin{prp}
  \label{prp:commutatorestimate1}
  If for all $d$-dimensional polyindices $\beta,b\in C_0(\N_0,\N_0^d)$
  with total order less than $k$, $\norm{\beta}+\norm{b}\le k$, the
  forms
  \begin{equation*}
    M_{(\jnorm{\cdot}^{2s_2}e^{-(\cdot)^2})^{\abs{\beta}}}T_{\beta,b}M_{(\jnorm{\cdot}^{2s_2}e^{-(\cdot)^2})^{\abs{\beta}}}
  \end{equation*}
  on $C_0^\infty(\R^d)$ extend to bounded operators on $\LL^2(\R^d)$,
  then so do the forms in \eqref{eq:commutatorformula} from
  Theorem~\ref{thm:commutatorformula} (with
  $M_{\alpha,\beta}T_{\beta,b}M_{a,b}$ interpreted as the bounded
  operator given by the form). Furthermore, we have the following
  estimate on $\ad_A^k(T_0)(H_0(\xi)+\ii)\inv$:
  \begin{equation}
    \label{eq:commutatorestimate1}
    \begin{split}
      \MoveEqLeft
      \norm{\ad_A^k(T_0)(H_0(\xi)+\ii)\inv}\\&\le\smashoperator{\sum_{\substack{\alpha,\beta,a,b
            \\
            \norm{\alpha}+\norm{\beta}+\norm{a}+\norm{b}=k}}}C_\xi''^k\frac{C^{\norm{\beta}+\norm{b}}k!}{c^{\abs{\gamma_{\beta+b}}}\gamma_{\beta+b}!
      }\norm{M_{(\jnorm{\cdot}^{2s_2}e^{-(\cdot)^2})^{\abs{\alpha}+\abs{\beta}}}T_{\beta,b}M_{(\omega_\xi+\ii)\inv(\jnorm{\cdot}^{2s_2}e^{-(\cdot)^2})^{\abs{a}+\abs{b}}}},
    \end{split}
  \end{equation}
  where $C_\xi''$ is the constant from
  Proposition~\ref{prp:estimatesMab}, which depends continuously on
  $\xi$, and $c$ and $C$ are the constants from Lemma~\ref{lem:redordfac}.
\end{prp}
\begin{proof}
  This follows easily from Theorem~\ref{thm:commutatorformula}, Theorem~\ref{cor:redordfac}, and Proposition~\ref{prp:estimatesMab}.
\end{proof}

\begin{thm}
  \label{thm:interaction}
 \leavevmode
  \begin{enumerate}
 \item
  Assume that for some $c>0$ and all pairs of polyindices $\beta,b$,
  \begin{equation}\label{eq:MTM}
    \norm{M_{(\jnorm{\cdot}^{2s_2}e^{-(\cdot)^2})^{\abs{\beta}}}T_{\beta,b}M_{(\omega_\xi+\ii)\inv(\jnorm{\cdot}^{2s_2}e^{-(\cdot)^2})^{\abs{b}}}}\le c^{\abs{\beta+b}}\gamma_{\beta,b}!\,,
  \end{equation}
  where $\gamma_{\beta,b}$ is given as in \eqref{eq:gammabetab}, and 
  $M_{(\jnorm{\cdot}^{2s_2}e^{-(\cdot)^2})^{\abs{\beta}}}T_{\beta,b}M_{(\omega_\xi+\ii)\inv(\jnorm{\cdot}^{2s_2}e^{-(\cdot)^2})^{\abs{b}}}$
  is interpreted as the bounded operator given by the form on
  $C_0^\infty(\R^d)$. Then there exists a constant $C>1$ such that
  \begin{equation*}
    \norm{\ad_A^k(T_0)(H_0(\xi)+\ii)\inv}\le (CC_\xi''c_d)^kk!\,
  \end{equation*}
  and the constant $C_\xi''$, which depends continuously on $\xi$,
  comes from Proposition~\ref{prp:estimatesMab} and $c_d$ comes from
  Proposition~\ref{prp:noofterms}. 
  
\item  In particular, \eqref{eq:MTM} holds, if there
  exists a constant $c>0$ such that for all $\beta,b$,
  \begin{align}\label{eq:T_sufficient}
    \norm{T_{\beta,b}(H_0(\xi)+\ii)\inv}\le c^{\abs{\beta+b}}\gamma_{\beta,b}!.
  \end{align}
\item Suppose that $V$ is as in either Condition \ref{cond:Lsw} or
  Condition \ref{cond:L1}. Then \eqref{eq:T_sufficient} is satisfied.
\end{enumerate}

\end{thm}

\begin{proof}
  The first part of the statement follows immediately from
  Remark~\ref{rem:redordfac}, Proposition~\ref{prp:noofterms}, and
  Proposition~\ref{prp:commutatorestimate1}. The second statement
  follows from observing that if
  $\norm{T_{\beta,b}(H_0(\xi)+\ii)\inv}<\infty$, then
  \begin{equation*}
    \norm{M_{(\jnorm{\cdot}^{2s_2}e^{-(\cdot)^2})^{\abs{\beta}}}T_{\beta,b}M_{(\omega_\xi+\ii)\inv(\jnorm{\cdot}^{2s_2}e^{-(\cdot)^2})^{\abs{b}}}}\le c_M^{\abs{\beta+b}}\norm{T_{\beta,b}(H_0(\xi)+\ii)\inv}
  \end{equation*}
  for some $C_M>0$. The third statement can be seen to be correct by
  the following argument. Let $\hat{V}$ satisfy Condition
  \ref{cond:Lsw}. We introduce the shorthand $j_p(k):=(1+\langle k
  \rangle^p)\inv$, where $p:=s_2$, see Condition~\ref{cond_ex}. By the
  weak Young inequality, see~\cite[p. 107]{LiebLoss},
  \begin{align*}
  \norm{T_{\beta,b}(H_0(\xi)+\ii)\inv} 
  &= 
  \sup_{\substack{\phi,\psi\in\cH\\ \|\phi\|=\|\psi\|=1}}|\langle \phi, T_{\beta,b}(H_0(\xi)+\ii)\inv \psi\rangle.\\
  &\leq 
  \|\hat{V}^{(\gamma_{\beta,b})}\|_{s,w} \sup_{\substack{\psi\in\cH\\ \|\psi\|=1}} \|j_p \psi \|_t,
  \end{align*}
where $\frac{1}{s} + \frac{1}{t} = \frac{3}{2}$. Due to $s\in(1,2)$, $t\in (1,2)$. Now by H\"older's inequality
\begin{align*}
\|j_p \psi \|_t \leq \|\psi\|_2 \|j_p^t\|_{\frac{2}{2-t}}^{\frac{1}{t}}
\end{align*}
which is finite, if $-p\frac{2t}{2-t}+d<0$. This however is equivalent to $p>d(1-\frac{1}{s})$ which is true by assumption. Thus,
\begin{align*}
  \norm{T_{\beta,b}(H_0(\xi)+\ii)\inv} \leq \|\hat{V}^{(\gamma_{\beta,b})}\|_{s,w}\|j_p^t\|_{\frac{2}{2-t}}^{\frac{1}{t}}
\end{align*}
and we see that the third statement follows if we assume
Condition~\ref{cond:Lsw}. If V satisfies Condition \ref{cond:L1}, the
proof is similar and can be carried out directly by applying (the
ordinary) Young inequality.
\end{proof}


\section{The Yukawa potential}\label{sec:Yukawa}

Before proving Proposition \ref{prp:yuk} we will introduce some notation. Let $d\in\N$ and $D_1,\dots,D_d$ be discs in $\C$ of radius $r$. We then define $D := D_1\times \cdots\times D_d$. We denote by $\Gamma$ the distinguished boundary of $D$, that is $\Gamma = \partial D_1\times \cdots\times \partial D_d$. Moreover, for any $\alpha \in\N^d$ and $z\in\mathds{C}^d$ we define $z^\alpha:=\prod_{j=1}^{d}z_j^{\alpha_j}$. For an analytic function $f : U\subset \C^d \rightarrow \C^d$ we denote by $f^{(\alpha)}$ the iterated partial derivatives of $f$ corresponding to the multi-index $\alpha$, that is $\alpha_j$ derivatives w.r.t. the $j$-th variable. If we denote by $\alpha+\mathbf{1}\in\N^d$ the multi-index with whose $j$-th coordinate is $\alpha_j+1$, the $d$-dimensional generalization of Cauchy's formula is then
\begin{align*}
f^{(\alpha)}(z)= \int_{\Gamma}\frac{f(w)}{(z-w)^{\alpha+\mathbf{1}}}\mathrm{d}^dw.
\end{align*}
Having taken care of these notational issues we can provide a proof of Proposition \ref{prp:yuk}.

\begin{proof}[Proof of Proposition~\ref{prp:yuk}]
	Clearly $\hat{V}$ has an extension to an analytic function into the 3-dimensional strip $S_1^3:=\{z\in\C^3~|~ |\Im(z_j)|< \tilde{r} \}$, where $\tilde{r}<1$. Hence, for $r\in(0,\tilde{r})$ and $k\in\R^3$ the 3-dimensional Cauchy formula allows us to estimate
	\begin{align*}
	|\hat{V}^{(\alpha)}(k)| 
	&\leq 
	\frac{\alpha!}{(2\pi)^3} \left|\int_{\Gamma} \frac{\hat{V}(z)}{(k-z)^{\alpha+\mathbf{1}}}\dd z \right| \leq
	\frac{\alpha!}{(2\pi)^3} \frac{1}{r^{|\alpha|}}\int\limits_0^{2\pi} \int\limits_0^{2\pi} \int\limits_0^{2\pi} \frac{4\pi}{1+|r-|k||^2} \dd t_1 \dd t_2 \dd t_3\\ 
	&= 
	\frac{\alpha!}{r^{|\alpha|}} \frac{4\pi}{1+|r-|k||^2},
	\end{align*}
	where $\Gamma$ is the distinguished boundary of the 3 dimensional polydisc of radius $r$. Let $\beta>0$. By the above computations we thus have
	\begin{align*}
	\{k\in\R^3~|~ |\hat{V}^{(\alpha)}(k)| >\beta\} 
	&\subset 
	\bigg\{k\in\R^3~\bigg|~ \frac{\alpha!}{r^{|\alpha|}} \frac{4\pi}{1+|r-|k||^2} >\beta\bigg\}\\
	&=
	\bigg\{k\in\R^3~\bigg|~ r-\left(\frac{1}{\beta_0}-1\right)^{\frac{1}{2}} < |k| < \left(\frac{1}{\beta_0}-1\right)^{\frac{1}{2}} + r \bigg\}\\
	&\subset 
	\bigg\{k\in\R^3~\bigg|~  |k| < \left(\frac{1}{\beta_0}-1\right)^{\frac{1}{2}} + r \bigg\},
		\end{align*}
	where $\beta_0 = \frac{r^{|\alpha|}}{\alpha!4\pi}\beta$. Due to $\{k\in\R^3~|~(1+|r-|k||^2)\inv>\beta\}=\emptyset$ for $\beta\geq1$, we can use the above inclusions to compute
	\begin{align*}
	\|\hat{V}^{(\alpha)} \|_{s,w} 
	&= 
	\sup_{\beta>0} \beta \left| \{k\in\R^3~|~ |\hat{V}^{(\alpha)}(k)| >\beta\} \right|^{\frac{1}{s}}\\
	&\leq
		4\pi\frac{\alpha!}{r^{|\alpha|}} \sup_{\beta\in(0,1)} \beta \left| \bigg\{k\in\R^3~\bigg|~ \frac{1}{1+|r-|z||^2} >\beta \bigg\} \right|^{\frac{1}{s}}\\
	&\leq
	4\pi\frac{\alpha!}{r^{|\alpha|}} \sup_{\beta\in(0,1)} \beta \left| \bigg\{k\in\R^3~\bigg|~ |k| < \left(\frac{1}{\beta}-1\right)^{\frac{1}{2}} + r \bigg\} \right|^{\frac{1}{s}}\\
	&=
	16\pi^2\frac{\alpha!}{r^{|\alpha|}} M_s,
	\end{align*}
	where 
	\begin{align*}
	M_s=\sup_{\beta\in(0,1)}  \beta^{1-\frac{3}{2s}}\left( \left(1-\beta\right)^{\frac{1}{2}} + r\beta^{\frac{1}{2}} \right)^{\frac{3}{s}} < \infty,
	\end{align*}
	due to $s>3/2$. Choosing $c= \max\{16\pi^2 M_s ,1\} r\inv$ completes the proof of the statement.
\end{proof}

\section*{Acknowledgments}
The authors would like to thank Jacob Schach M\o ller for useful
discussions. M.~E.\ acknowledges the support of the Lundbeck
Foundation and the German Research Foundation (DFG) through the
Graduiertenkolleg 1838 and M.G.R. acknowledges support from the Danish
Council for Independent Research | Natural Sciences, grant 12-124675,
"Mathematical and Statistical Analysis of Spatial Data".

\bibliographystyle{amsplain} \bibliography{ER}

\end{document}